\documentclass[11pt]{article}
\usepackage[left=1in,top=1in,right=1in,bottom=1in]{geometry}
\usepackage{times}

\usepackage{verbatim}
\usepackage{amsmath}
\usepackage{amssymb}
\usepackage{amsthm}
\usepackage{rotating}
\usepackage{algorithm}
\usepackage[noend]{algpseudocode}

\usepackage{silence}
\usepackage{tabularx}
\usepackage{graphicx}
\usepackage{url}
\usepackage{multirow}
\usepackage{subfig}

\usepackage[noend]{algpseudocode}
\usepackage{siunitx}

\newtheorem{theorem}{Theorem}
\newtheorem{definition}{Definition}

\newtheorem{proposition}{Proposition}
\DeclareMathOperator*{\argmax}{arg\,max}
\DeclareMathOperator*{\argmin}{arg\,min}

\allowdisplaybreaks

\begin{document}
\title{Safe Equilibrium}
\author{Sam Ganzfried\\
Ganzfried Research\\
sam.ganzfried@gmail.com
}

\date{\vspace{-5ex}}

\maketitle

\begin{abstract}
The standard game-theoretic solution concept, Nash equilibrium, assumes that all players behave rationally. If we follow a Nash equilibrium and opponents are irrational (or follow strategies from a different Nash equilibrium), then we may obtain an extremely low payoff. On the other hand, a maximin strategy assumes that all opposing agents are playing to minimize our payoff (even if it is not in their best interest), and ensures the maximal possible worst-case payoff, but results in exceedingly conservative play. We propose a new solution concept called safe equilibrium that models opponents as behaving rationally with a specified probability and behaving potentially arbitrarily with the remaining probability.  We prove that a safe equilibrium exists in all strategic-form games (for all possible values of the rationality parameters), and prove that its computation is PPAD-hard. We present exact algorithms for computing a safe equilibrium in both 2 and $n$-player games, as well as scalable approximation algorithms.
\end{abstract}

\section{Introduction}
\label{se:intro}
In designing a strategy for a multiagent interaction an agent must balance between the assumption that opponents are behaving rationally with the risks that may occur if opponents behave irrationally. Most classic game-theoretic solution concepts, such as Nash equilibrium (NE), assume that all players are behaving rationally (and that this fact is common knowledge).  
On the other hand, a maximin strategy plays a strategy that has the largest worst-case guaranteed expected payoff; this limits the potential downside against a worst-case and potentially irrational 
opponent, but can also cause us to achieve significantly lower payoff against rational opponents. In two-player zero-sum games, Nash equilibrium and maximin strategies are equivalent (by the minimax theorem),
and these two goals are completely aligned. But in non-zero-sum games and games with more than two players, this is not the case. In these games we can potentially obtain arbitrarily low payoff
by following a Nash equilibrium strategy, but if we follow a maximin strategy will likely be playing far too conservatively. While the assumption that opponents are exhibiting a degree of rationality, as well as the desire to limit worst-case performance in the case of irrational opponents, are both desirable, neither the Nash equilibrium nor maximin solution concept is definitively compelling on its own.

We propose a new solution concept that balances between these two extremes. In a two-player general-sum game, we define an \emph{$\epsilon$-safe equilibrium} ($\epsilon$-SE) as a strategy profile where each player $i$ is playing a strategy that minimizes performance of the opponent with probability $\epsilon_i$, and is playing a best response to the opponent's strategy with probability $1-\epsilon_i$, where $\epsilon = (\epsilon_1, \epsilon_2).$ As a special case, if we are interested in constructing a strategy for player 1, we can set $\epsilon_1 = 0$, assuming irrationality just for player 2. We can generalize this to an $n$-player game by assuming that all players $i \neq 1$ are playing a strategy that minimizes player 1's expected payoff with probability $\epsilon_i$, and are playing a best response to all other players' strategies with probability $1-\epsilon_i$, while player 1 plays a best response to all other players' strategies. This concept balances explicitly between the assumption of players' rationality and the desire to ensure safety in the worst case through the $\epsilon_i$ parameters. From a theoretical perspective we show that an $\epsilon$-safe equilibrium is always guaranteed to exist and is PPAD-hard to compute (assuming $\epsilon_i < 1).$ Thus, it has the same existence and complexity results as Nash equilibrium.

Several other game-theoretic solution concepts have been previously proposed to account for degrees of opponents' rationality. The most prominent is \emph{trembling-hand perfect equilibrium} (THPE), which is a 
refinement of Nash equilibrium that is robust to the possibility that players ``tremble'' and play each pure strategy with arbitrarily small probability~\cite{Selten75:Reexamination}.  The concept of $\epsilon$-safe equilibrium differs from THPE in several key ways. First, it allows a player to specify an arbitrary belief on the probability that each other player is irrational, rather than assume that it is an extremely small value. In domains like national security or driving we risk losing lives in the event that we fail to properly account for opponents' irrationality, and may elect to use larger values for $\epsilon_i$ than in situations where safety is less of a concern. In an $\epsilon$-SE a player can specify the values for $\epsilon_i$ based on prior beliefs about the opponent or any relevant domain-specific knowledge, and is still free to use values that are extremely close to 0 as in THPE. Furthermore, a THPE is a refinement of NE, while $\epsilon$-SE and NE are incomparable (an $\epsilon$-SE may not be an NE and vice versa). Another related concept is that of a \emph{safe strategy} and \emph{$\epsilon$-safe strategy}~\cite{McCracken04:Safe}. A strategy for a player in a two-player zero-sum game is called safe if it guarantees an expected payoff of at least $v^*$---the value of the game to the player---in the worse case. Note that this also coincides with the set of minimax, maximin, and Nash equilibrium strategies. A strategy is $\epsilon$-safe if it obtains a worst-case expected payoff of at least $v^*-\epsilon.$ The concepts of safe and $\epsilon$-safe strategies are defined just for two-player zero-sum games, while safe and $\epsilon$-safe equilibrium also apply to non-zero-sum and multiplayer games. 

We note that a belief of opponents' ``irrationality'' does not necessarily indicate that we believe them to be ``stupid'' or ``crazy.'' It may simply correspond to a belief that the opponent may have a different model of the game than we do. For example, our analysis may indicate that a successful attack on a location would result in a certain payoff for the opponent, while their analysis indicates a different payoff. In addition to potentially constructing different assessments of their own or other players' payoffs, opponents may also be ``irrational'' because they are using an algorithm for computing a Nash equilibrium that is only able to yield an approximation, or just a different Nash equilibrium from what other players have calculated (in fact, these cases do not actually seem to be irrational at all, since computing a Nash equilibrium is computationally challenging and many games have multiple Nash equilibria). If any of these situations arise, then simply following an arbitrary Nash equilibrium strategy runs a risk of an extremely low payoff, and there is potential for significant benefit by ensuring a degree of safety.

An alternative approach for modeling potentially irrational opponents is to incorporate an \emph{opponent modeling algorithm}. Opponent modeling algorithms typically require the use of domain-specific expertise and databases of historical play to construct a prior distribution for opponents' strategies and use machine learning algorithms that predict a strategy (or distribution of strategies) for the opponents taking into account the prior and observations of publicly observable actions. This can be extremely valuable if domain expertise, large amounts of historical data, and a large number of observations of opponents' play are available. Often such information is not available and we are forced to construct our strategy without any additional data-specific tendencies of the opponent. We note that if such data is available, the safe equilibrium concept can be integrated with opponent modeling to successfully achieve robust opponent exploitation. An approach called a restricted Nash response was developed for two-player zero-sum games where the opponent is restricted to play a fixed strategy $\sigma_{\mbox{fix}}$ determined by an opponent model with probability $p$ and plays a best response to us with probability $1-p$ while we best respond to the opponent (it is shown that this approach is equivalent to playing an $\epsilon$-safe best response to $\sigma_{\mbox{fix}}$ (a best response to $\sigma_{\mbox{fix}}$ out of strategies that are $\epsilon$-safe) for some $\epsilon$)~\cite{Johanson07:Computing}. It was shown that for certain values of $p$ this approach can result in a significant reduction in the level of exploitability of our own strategy while only a slight reduction in our degree of exploitation of the opponent's strategy. It has also been shown that approaches that compute an $\epsilon$-safe best response to a model of the opponent's strategy for dynamically changing values of $\epsilon$ in repeated two-player zero-sum games can guarantee safety~\cite{Ganzfried15:Safe}. An $\epsilon$-safe equilibrium strategy can be used in non-zero-sum and multiplayer games where models are available for the opponents' strategies by assuming each opponent $i$ follows their opponent model with probability $\epsilon_i$ instead of playing a worst-case strategy for us, while also playing a best response with probability $1-\epsilon_i.$ Thus, in the event that an opponent model is available we can view safe equilibrium as a generalization of restricted Nash response to achieve robust opponent exploitation in the settings of non-zero-sum and multiplayer games.

\section{Safe Equilibrium}
\label{se:se}
A \emph{strategic-form game} consists of a finite set of players $N = \{1,\ldots,n\}$, a finite set of pure strategies $S_i$ for each player $i \in N$, and a real-valued utility for each player for each strategy vector (aka \emph{strategy profile}), $u_i : \times_i S_i \rightarrow \mathbb{R}$. A \emph{mixed strategy} $\sigma_i$ for player $i$ is a probability distribution over pure strategies, where $\sigma_i(s_{i'})$ is the probability that player $i$ plays pure strategy $s_{i'} \in S_i$ under $\sigma_i$. Let $\Sigma_i$ denote the full set of mixed strategies for player $i$. A strategy profile $\sigma^* = (\sigma^*_1,\ldots,\sigma^*_n)$ is a \emph{Nash equilibrium} if $u_i(\sigma^*_i,\sigma^*_{-i}) \geq u_i(\sigma_i, \sigma^*_{-i})$ for all $\sigma_i \in \Sigma_i$ for all $i \in N$, where $\sigma^*_{-i} \in \Sigma_{-i}$ denotes the vector of the components of strategy $\sigma^*$ for all players excluding player $i$. Here $u_i$ denotes the expected utility for player $i$, and $\Sigma_{-i}$ denotes the set of strategy profiles for all players excluding player $i$. 

A mixed strategy $\sigma^*_i$ for player $i$ is a \emph{maximin strategy} if 
$$\sigma^*_i \in \argmax_{\sigma_i \in \Sigma_i} \min_{\sigma_{-i} \in \Sigma_{-i}} u_i(\sigma_i,\sigma_{-i}).$$

\begin{definition}
\label{de:se2p}
Let $G$ be a two-player strategic-form game. Let $\epsilon = (\epsilon_1,\epsilon_2)$, where $\epsilon_i \in [0,1]$ for $i = 1,2.$ A strategy profile $\sigma^*$ is an \emph{$\epsilon$-safe equilibrium} if there exist mixed strategies $\tau^*_i, \rho^*_i \in \Sigma_i$ where $\sigma^*_i = \epsilon_i \tau^*_i + (1-\epsilon_i)\rho^*_i$ for $i = 1,2$ such that $\rho^*_i \in \argmax_{\sigma_i \in \Sigma_i} u_i(\sigma_i,\sigma^*_{-i})$,
$\tau^*_i \in \argmin_{\sigma_i \in \Sigma_i} u_{-i}(\sigma^*_{-i},\sigma_i)$.
\end{definition}

In practice player $i$ would likely want to set $\epsilon_i = 0$ and $\epsilon_j > 0$ for $j \neq i$ when determining their own strategy, though Definition~\ref{de:se2p} allows an arbitrary value of $\epsilon_i \in [0,1]$ as well. It may make sense for player $i$ to set $\epsilon_i > 0$ if they believe both that the opponent is irrational with some probability $\epsilon_{-i}$, and if they also believe that the opponent believes that player $i$ is irrational with some probability $\epsilon_i.$

\begin{theorem}
\label{th:2p}
Let $G = (N,(S_i)_{i \in N},(u_i)_{i \in N})$ be a two-player strategic-form game, and let $\epsilon = (\epsilon_1,\epsilon_2)$, where $\epsilon_1,\epsilon_2 \in [0,1].$ Then $G$ contains an $\epsilon$-safe equilibrium.
\end{theorem}
\begin{proof}
Define $G' = (N',(S'_i)_{i \in N},(u'_i)_{i \in N})$ to be the following game. $N' = \{1,2,3,4\}$, $S'_1 = S'_2 = S_1$, $S'_3 = S'_4 = S_2.$ 
For $s'_i \in S'_i$, define $u'_i$ as follows for $i \in N$:
$$u'_1(s'_1,s'_2,s'_3,s'_4) = -\epsilon_2 u_2(s'_1,s'_3) - (1-\epsilon_2)u_2(s'_1,s'_4)$$
$$u'_2(s'_1,s'_2,s'_3,s'_4) = \epsilon_2 u_1(s'_2,s'_3) + (1-\epsilon_2)u_1(s'_2,s'_4)$$
$$u'_3(s'_1,s'_2,s'_3,s'_4) = -\epsilon_1 u_1(s'_1,s'_3) - (1-\epsilon_1)u_1(s'_2,s'_3)$$
$$u'_4(s'_1,s'_2,s'_3,s'_4) = \epsilon_1 u_2(s'_1,s'_4) + (1-\epsilon_1)u_2(s'_2,s'_4)$$
Player 1's strategy corresponds to $\tau^*_1$, player 2's strategy corresponds to $\rho^*_1$, player 3's strategy corresponds to $\tau^*_2$, and player 4's strategy corresponds to $\rho^*_2.$
By Nash's existence theorem, the game $G'$ has a Nash equilibrium, which corresponds to an $\epsilon$-safe equilibrium of $G.$
\end{proof}

\begin{theorem}
\label{th:complexity-2p}
Let $\epsilon = (\epsilon_1,\epsilon_2)$, where $\epsilon_1,\epsilon_2 \in [0,1)$ are fixed constants. The problem of computing an $\epsilon$-safe equilibrium is PPAD-hard.
\end{theorem}
\begin{proof}
Let $G = (N,(S_i)_{i \in N},(u_i)_{i \in N})$ be a two-player strategic-form game. 
Suppose that $k$ is the smallest possible payoff for any player in $G$, and let $k' = k-1.$
Define the game $G' = (N',(S'_i)_{i \in N},(u'_i)_{i \in N})$ as follows. $N' = \{1,2\}$, $S'_1 = S_1 \cup t$, $S'_2 = S_2 \cup t.$ 
For $s'_i \in S'_i$, define $u'_i$ as follows for $i \in N$:
$$u'_i(s'_1,s'_2) = u_i(s'_1,s'_2) \mbox{ for } s_1 \in S_1, s_2 \in S_2.$$
$$u'_i(t,s'_2) = k' \mbox{ for } s'_2 \in S_2.$$
$$u'_i(s'_1,t) = k' \mbox{ for } s'_1 \in S_1.$$
$$u'_i(t,t) = k'-1.$$

Suppose we can efficiently compute an $\epsilon$-safe equilibrium of $G'$, denoted by $\sigma^{G'}.$  
Then we have $\sigma^{G'}_i = \epsilon_i \tau^*_i + (1-\epsilon_i)\rho^*_i$ for $i = 1,2,$ where $\rho^*_i \in \argmax_{\sigma'_i \in \Sigma'_i} u_i(\sigma'_i,\sigma^{G'}_{-i}),$ $\tau^*_i \in \argmin_{\sigma'_i \in \Sigma'_i} u_{-i}(\sigma^{G'}_{-i},\sigma'_i).$ 

I claim that $\rho^*$ is a Nash equilibrium of $G$. First note that $\rho^*_i$ must put probability 0 on $t$ for all players, since $t$ is strictly dominated. So it is a valid strategy profile of $G$. Also note that $\tau^*_i$ must put probability 1 on $t$ for all $i.$

Suppose that player $i$ can improve performance in $G$ by deviating to $\eta_i.$ Then 
$$u_i(\eta_i,\rho^*_{-i}) > u_i(\rho^*_i,\rho^*_{-i})$$
$$(1-\epsilon_{-i}) u_i(\eta_i,\rho^*_{-i}) + \epsilon_{-i} k' > (1-\epsilon_{-i}) u_i(\rho^*_i,\rho^*_{-i}) + \epsilon_{-i} k'$$
$$(1-\epsilon_{-i}) u_i(\eta_i,\rho^*_{-i}) + \epsilon_{-i} u_i(\eta_i,t) > (1-\epsilon_{-i}) u_i(\rho^*_i,\rho^*_{-i}) +  \epsilon_{-i} u_i(\eta_i,t)$$
$$(1-\epsilon_{-i}) u_i(\eta_i,\rho^*_{-i}) + \epsilon_{-i} u_i(\eta_i,t) > (1-\epsilon_{-i}) u_i(\rho^*_i,\rho^*_{-i}) +  \epsilon_{-i} u_i(\rho^*_i,t)$$
$$(1-\epsilon_{-i}) u_i(\eta_i,\rho^*_{-i}) + \epsilon_{-i} u_i(\eta_i,\tau^*_{-i}) > (1-\epsilon_{-i}) u_i(\rho^*_i,\rho^*_{-i}) +  \epsilon_{-i} u_i(\rho^*_i,\tau^*_{-i})$$
$$u_i(\eta_i,\sigma^{G'}_{-i}) > u_i(\rho^*_i,\sigma^{G'}_{-i}).$$
This contradicts the fact that $\rho^*_i \in \argmax_{\sigma'_i \in \Sigma'_i} u_i(\sigma'_i,\sigma^{G'}_{-i}).$ So we have a contradiction, and conclude that no player can improve performance in $G$ by deviating from $\rho^*.$ So $\rho^*$ is a Nash equilibrium of $G$.

Since the problem of computing a Nash equilibrium is PPAD-hard and we have reduced it to the problem of computing an $\epsilon$-safe equilibrium, this shows that the problem of computing an $\epsilon$-safe equilibrium is PPAD-hard.
\end{proof}

For $n > 2$ players, we designate one of the players as being a special player, say player 1. We can view player 1 as representing ``ourselves'' as a decision-making agent, and the other players as unpredictable opponents. Player 1 then best responds to the strategy profile of all other players, while each opposing player $i$ mixes between playing a strategy that minimizes player 1's payoff and a strategy that maximizes player $i$'s payoff in response to the strategy profile of the other players.

\begin{definition}
\label{de:senp}
Let $G$ be an n-player strategic-form game. Let $\epsilon = (\epsilon_2,\ldots,\epsilon_n)$, where $\epsilon_i \in [0,1].$ A strategy profile $\sigma^*$ is an \emph{$\epsilon$-safe equilibrium} if there exists a mixed strategy $\sigma^*_1$ for player 1 and mixed strategies $\tau^*_i, \rho^*_i \in \Sigma_i$ where $\sigma^*_i = \epsilon_i \tau^*_i + (1-\epsilon_i)\rho^*_i$ for $i = 2,\ldots,n$ such that $\rho^*_i \in \argmax_{\sigma_i \in \Sigma_i} u_i(\sigma_i,\sigma^*_{-i})$,
$\tau^*_i \in \argmin_{\sigma_i \in \Sigma_i} u_{1}(\sigma^*_1,\sigma')$,
$\sigma^*_1 \in \argmax_{\sigma_1 \in \Sigma_1} u_1(\sigma_1,\sigma^*_{-1}),$
where $\sigma'$ is the strategy profile for players in $\{2,\ldots,n\}$ where player $i$ plays $\sigma_i$ and the other players $j \neq i$ play $\sigma^*_j.$
\end{definition}

The proof of Theorem~\ref{th:2p} extends naturally to $n > 2$ players as well by creating a $2(n-1)+1 = 2n-1$ player game with 2 new players corresponding to each player in the initial game for $i > 1$, plus player 1.

\begin{theorem}
\label{th:np}
Let $G = (N,(S_i)_{i \in N},(u_i)_{i \in N})$ be an $n$-player strategic-form game, and let $\epsilon = (\epsilon_2,\ldots,\epsilon_n)$, where $\epsilon_i \in [0,1].$ Then $G$ contains an $\epsilon$-safe equilibrium.
\end{theorem}

The proof of Theorem~\ref{th:complexity-2p} also straightforwardly extends to $n$ players.

\begin{theorem}
\label{th:complexity-np}
Let $\epsilon = (\epsilon_2,\ldots,\epsilon_n)$, where $\epsilon_i \in [0,1)$ are fixed constants. The problem of computing an $\epsilon$-safe equilibrium is PPAD-hard.
\end{theorem}

As an example, consider the classic game of Chicken, with payoffs given by Figure~\ref{fi:chicken}. The first action for each player corresponds to the ``swerve'' action, while the second corresponds to the ``straight'' action.
\begin{quotation}
The game of chicken models two drivers, both headed for a single-lane bridge from opposite directions. The first to swerve away yields the bridge to the other. If neither player swerves, the result is a costly deadlock in the middle of the bridge, or a potentially fatal head-on collision. It is presumed that the best thing for each driver is to stay straight while the other swerves (since the other is the ``chicken'' while a crash is avoided). Additionally, a crash is presumed to be the worst outcome for both players. This yields a situation where each player, in attempting to secure their best outcome, risks the worst~\cite{Wiki21:Chicken}.
\end{quotation}

\begin{figure}
\begin{equation*}
\begin{bmatrix}
(0,0) & (-1,+1) \\
(+1,-1) & (-10,-10) \\
\end{bmatrix}
\end{equation*}
\caption{Payoff matrix for game of Chicken.}
\label{fi:chicken}
\end{figure}

The unique mixed-strategy Nash equilibrium $\sigma^{NE}$ in the Chicken game is for each player to swerve with probability 0.9 (there are also two pure-strategy equilibria where one player swerves and the other player doesn't), and the unique maximin strategy $\sigma^M$ is to swerve with probability 1. If we set $\epsilon_1 = 0$, then it turns out that $\sigma^{NE}_1$ is an $\epsilon$-safe equilibrium strategy for player 1 for $0 \leq \epsilon_2 \leq 0.1$, and $\sigma^{M}_1$ is an $\epsilon$-safe equilibrium strategy for player 1 for $0.1 \leq \epsilon_2 \leq 1.$ 
It is not necessary that an $\epsilon$-safe equilibrium strategy always corresponds to a Nash equilibrium or maximin strategy. For example, with $\epsilon_1 = 0.05$ and $\epsilon_2 = 0.15$, an $\epsilon$-safe equilibrium strategy profile is for player 1 to swerve with probability 0.95 and player 2 to swerve with probability 0. 

As another example, consider the security game depicted in Figure~\ref{fi:security}, where the row player selects one of three targets to defend while the column player selects a target to attack. A Nash equilibrium for player 1 (row player) $\sigma^{NE}_1$ is to defend the targets with probabilities $(0.3136,0.4661,0.2203)$, and a maximin strategy $\sigma^{M}_1$ is to defend the targets with probabilities $(0.6144,0.0131,0.3725).$ Again using $\epsilon_1 = 0$, for  $\epsilon_2 \in [0,0.314]$ it turns out that $\sigma^{NE}_1$ is an $\epsilon$-safe equilibrium strategy for player 1, and for $\epsilon_2 \in [0.569,1]$ $\sigma^{M}_1$ is an $\epsilon$-safe equilibrium strategy for player 1. But for the region $\epsilon_2 \in [0.314,0.569]$ it turns out that the strategy $(0.4437,0.3666,0.1897)$ is an $\epsilon$-safe equilibrium strategy for player 1, which is neither a Nash equilibrium strategy nor a maximin strategy.

\begin{figure}
\begin{equation*}
\begin{bmatrix}
(4,-3) &(-1,1) &(-7,2) \\
(-5,5) &(2,-1) &(-1,4) \\
(-9,1) &(-1,8) &(9,-4) \\ 
\end{bmatrix}
\end{equation*}
\caption{Security game payoff matrix.}
\label{fi:security}
\end{figure}

\clearpage
\section{Algorithms for computing safe equilibrium}
We first present an exact algorithm for computing an $\epsilon$-safe equilibrium, followed by an approximation algorithm that runs quickly on large instances. The exact algorithm is based on a mixed-integer feasibility program formulation. We first present the algorithm for two players, for arbitrary $\epsilon_i \in [0,1].$ The algorithm builds on a related linear mixed-integer feasibility program formulation for computing Nash equilibrium in two-player general-sum games~\cite{Sandholm05:Mixed}.

We quote from the original description of the program formulation for two-player Nash equilibrium, and present the formulation below:

\begin{quote}
In our first formulation, the feasible solutions are exactly the equilibria of the game.  For every pure strategy $s_i$, there is  binary variable $b_{s_i}$.  If this variable is set to 1, the probability placed on the strategy must be 0. If it is set to 0, the strategy is allowed to be in the support, but the regret of the strategy must be 0. The formulation has the following variables other than the $b_{s_i}$.  For each player, there is a variable $u_i$ indicating the highest possible expected utility that that player can obtain given the other player's mixed strategy. For every pure strategy $s_i$, there is a variable $p_{s_i}$ indicating the probability placed on that strategy, a variable $u_{s_i}$ indicating the expected utility of playing that strategy (given the other player's mixed  strategy), and a variable $r_{s_i}$ indicating the regret of playing $s_i$.  The constant $U_i$ indicates the maximum difference between two utilities in the game for player $i$: 
$U_i = \max_{s^h_i, s^l_i \in S_i, s^h_{1-i},s^l_{1-i} \in S_{1-i}} \left[ u_i(s^h_i, s^h_{1-i}) - u_i(s^l_i, s^l_{1-i}) \right].$ 
The formulation follows below~\cite{Sandholm05:Mixed}.
\end{quote}

Find $p_{s_i},u_i,u_{s_i},r_{s_i},b_{s_i}$ such that:
\begin{align}
&\sum_{s_i \in S_i} p_{s_i} = 1 \mbox{ for all } i \\
&u_{s_i} = \sum_{s_{1-i} \in S_{1-i}} p_{s_{1-i}}u_i(s_i,s_{1-i}) \mbox{ for all } i, s_i \in S_i \\
&u_i \geq u_{s_i} \mbox{ for all } i, s_i \in S_i \label{eq:2p-u} \\
&r_{s_i} = u_i - u_{s_i} \mbox{ for all } i, s_i \in S_i \label{eq:2p-r1}\\
&p_{s_i} \leq 1 - b_{s_i} \mbox{ for all } i, s_i \in S_i \label{eq:2p-p}\\
&r_{s_i} \leq U_i b_{s_i} \mbox{ for all } i, s_i \in S_i \label{eq:2p-r2}\\
&p_{s_i} \geq 0 \mbox{ for all } i, s_i \in S_i \\
&u_i \geq 0 \mbox{ for all } i \\
&u_{s_i} \geq 0 \mbox{ for all } i, s_i \in S_i \\
&r_{s_i} \geq 0 \mbox{ for all } i, s_i \in S_i \label{eq:2p-rbound}\\
&b_{s_i} \mbox{ binary in } \{0,1\} \mbox{ for all } i, s_i \in S_i 
\end{align}

\begin{quote}
The first four constraints ensure that the $p_{s_i}$ values constitute a valid probability distribution and define the regret of a strategy. Constraint~\ref{eq:2p-p} ensures that $b_{s_i}$ can be set to 1 only when no probability is placed on $s_i$. On the other hand, Constraint~\ref{eq:2p-r2} ensures that the regret of a strategy equals 0, unless $b_{s_i} = 1$,  in which  case  the constraint is vacuous because the regret can never exceed $U_i$. (Technically, Constraint~\ref{eq:2p-u} is redundant as it follows from Constraints~\ref{eq:2p-r1} and~\ref{eq:2p-rbound}.)~\cite{Sandholm05:Mixed}
\end{quote}

We modify this program as follows. For every pure strategy $s_i$, we include two binary variables $b^1_{s_i}, b^2_{s_i}$. The first one corresponds to player $i$'s best response strategy, and the second corresponds to player $i$'s strategy that minimizes the opponent's payoff. Additionally, we include variables $u^j_i,p^j_{s_i},u^j_{s_i},r^j_{s_i}$ and constants $U^j_i$ for $j = 1,2.$ Given constants $\epsilon_1,\epsilon_2 \in [0,1]$, we create the following formulation for computing an $\epsilon$-safe equilibrium (note that we have removed the redundant Constraint~\ref{eq:2p-u}).

Find $p^j_{s_i},u^j_i,u^j_{s_i},r^j_{s_i},b^j_{s_i}$ such that:
\begin{align*}
&\sum_{s_i \in S_i} p^j_{s_i} = 1 \mbox{ for all } i,j \\
&u^1_{s_i} = \sum_{s_{1-i} \in S_{1-i}} \left[ u_i(s_i,s_{1-i})\left(\epsilon_{1-i} p^2_{s_{1-i}} + (1-\epsilon_{1-i})p^1_{s_{1-i}}\right)\right] \mbox{ for all } i, s_i \in S_i \\
&u^2_{s_i} = \sum_{s_{1-i} \in S_{1-i}} \left[-u_{-i}(s_{1-i},s_i)\left(\epsilon_{1-i} p^2_{s_{1-i}} + (1-\epsilon_{1-i})p^1_{s_{1-i}}\right)\right] \mbox{ for all } i, s_i \in S_i \\
&r^j_{s_i} = u^j_i - u^j_{s_i} \mbox{ for all } i, s_i \in S_i, j \\
&p^j_{s_i} \leq 1 - b^j_{s_i} \mbox{ for all } i, s_i \in S_i, j \\
&r^j_{s_i} \leq U^j_i b^j_{s_i} \mbox{ for all } i, s_i, j \in S_i \\
&p^j_{s_i} \geq 0 \mbox{ for all } i, s_i \in S_i, j \\
&u^j_i \geq 0 \mbox{ for all } i,j \\
&u^j_{s_i} \geq 0 \mbox{ for all } i, s_i \in S_i, j \\
&r^j_{s_i} \geq 0 \mbox{ for all } i, s_i \in S_i, j \\
&b^j_{s_i} \mbox{ binary in } \{0,1\} \mbox{ for all } i, s_i \in S_i, j
\end{align*}

For three players, we can use the following formulation, where new variables $p^{j_i,j_k}_{s_i,s_k}$ 
denote the product of the variables $p^{j_i}_{s_i}$ and $p^{j_k}_{s_k}.$ For the special player 1 we just have superscript 1, while for players 2 and 3 we have superscripts 1 and 2 (corresponding to the best-response strategy and the strategy that is worst-case for player 1). This formulation can be straightforwardly extended to a non-convex quadratically-constrained mixed-integer feasibility program formulation for $n$ players, and is based on a recent algorithm for computing multiplayer Nash equilibrium~\cite{Ganzfried20:Fast}. 

Find $p^j_{s_i},u^j_i,u^j_{s_i},r^j_{s_i},b^j_{s_i},p^{j_i,j_k}_{s_i,s_k}$ subject to:
\begin{align*}
&\sum_{s_i \in S_i} p^j_{s_i} = 1 \mbox{ for all } i, j \\
&u^1_{s_1} = \sum_{s_2 \in S_2} \sum_{s_3 \in S_3} [u_1(s_1,s_2,s_3) (\epsilon_2 \epsilon_3 p^{2,2}_{s_2,s_3} + \epsilon_2(1-\epsilon_3)p^{2,1}_{s_2,s_3}\\
& + (1-\epsilon_2)\epsilon_3p^{1,2}_{s_2,s_3} + (1-\epsilon_2)(1-\epsilon_3)p^{2,2}_{s_2,s_3})] \mbox{ for all } s_1 \in S_1\\
&u^1_{s_2} = \sum_{s_1 \in S_1} \sum_{s_3 \in S_3} [u_2(s_1,s_2,s_3) (\epsilon_3 p^{1,2}_{s_1,s_3} + (1-\epsilon_3)p^{1,1}_{s_1,s_3})] \mbox{ for all } s_2 \in S_2\\
&u^2_{s_2} = \sum_{s_1 \in S_1} \sum_{s_3 \in S_3} [-u_1(s_1,s_2,s_3) (\epsilon_3 p^{1,2}_{s_1,s_3} + (1-\epsilon_3)p^{1,1}_{s_1,s_3})] \mbox{ for all } s_2 \in S_2\\
&u^1_{s_3} = \sum_{s_1 \in S_1} \sum_{s_2 \in S_2} [u_3(s_1,s_2,s_3) (\epsilon_2 p^{1,2}_{s_1,s_2} + (1-\epsilon_2)p^{1,1}_{s_1,s_2})] \mbox{ for all } s_3 \in S_3\\
&u^2_{s_3} = \sum_{s_1 \in S_1} \sum_{s_2 \in S_2} [-u_1(s_1,s_2,s_3) (\epsilon_2 p^{1,2}_{s_1,s_2} + (1-\epsilon_2)p^{1,1}_{s_1,s_2})] \mbox{ for all } s_3 \in S_3\\
&p^{j_i,j_k}_{s_i,s_k} = p^{j_i}_{s_i} \cdot p^{j_k}_{s_k} \mbox{ for all } j_i, j_k, s_i \in S_1, s_k \in S_2\\
&p^{j_i,j_k}_{s_i,s_k} = p^{j_i}_{s_i} \cdot p^{j_k}_{s_k} \mbox{ for all } j_i, j_k, s_i \in S_1, s_k \in S_3\\
&p^{j_i,j_k}_{s_i,s_k} = p^{j_i}_{s_i} \cdot p^{j_k}_{s_k} \mbox{ for all } j_i, j_k, s_i \in S_2, s_k \in S_3\\
&r^j_{s_i} = u^j_i - u^j_{s_i} \mbox{ for all } i, s_i \in S_i, j \\
&p^j_{s_i} \leq 1 - b^j_{s_i} \mbox{ for all } i, s_i \in S_i, j \\
&r^j_{s_i} \leq U^j_i b^j_{s_i} \mbox{ for all } i, s_i, j \in S_i \\
&p^j_{s_i} \geq 0 \mbox{ for all } i, s_i \in S_i, j \\
&u^j_i \geq 0 \mbox{ for all } i,j \\
&u^j_{s_i} \geq 0 \mbox{ for all } i, s_i \in S_i, j \\
&r^j_{s_i} \geq 0 \mbox{ for all } i, s_i \in S_i, j \\
&b^j_{s_i} \mbox{ binary in } \{0,1\} \mbox{ for all } i, s_i \in S_i, j
\end{align*}

This provides an exact algorithm for computing $\epsilon$-safe equilibrium in $n$-player games. Next we consider an approximation algorithm that scales to large games. Two algorithms that have been recently applied to approximate Nash equilibrium in large multiplayer games are (counterfactual) regret minimization~\cite{Zinkevich07:Regret} and fictitious play~\cite{Brown51:Iterative,Robinson51:Iterative}. These are iterative self-play procedures that have been proven to converge to Nash equilibrium in two-player zero-sum games, but not for more than two players. Recently it has been shown that fictitious play outperforms regret minimization for multiplayer games~\cite{Ganzfried20:Fictitious}, so we will base our algorithms on fictitious play. Algorithm~\ref{al:approx-2plr} presents our algorithm for computing $\epsilon$-safe equilibrium in two-player games, and Algorithm~\ref{al:approx-nplr} presents our algorithm for $n$-player games. Note that $\rho^t_i$ and $\tau^t_i$ are not actually needed in the algorithms (for $t > 0$), but they will be useful for evaluating the algorithms in our experiments.

\begin{algorithm}
\caption{Approximation algorithm for $\epsilon$-safe equilibrium in two-player games}
\label{al:approx-2plr}
\textbf{Inputs}: Game $G$, $\epsilon_1,\epsilon_2 \in [0,1]$, initial mixed strategies $\tau^0_i, \rho^0_i \in \Sigma_i$ for $i = 1,2$, number of iterations $T$.
\begin{algorithmic}
\State $\sigma^0_i = \epsilon_i \tau^0_i + (1-\epsilon_i) \rho^0_i$ for $i = 1,2$
\For {$t = 1$ to $T$}
\State $\rho'_i = \argmax_{\sigma_i \in \Sigma_i} u_i(\sigma_i,\sigma^{t-1}_{-i})$ for $i = 1,2$
\State $\tau'_i = \argmin_{\sigma_i \in \Sigma_i} u_{-i}(\sigma^{t-1}_{-i},\sigma_i)$ for $i = 1,2$
\State $\sigma'_i = \epsilon_i \tau'_i + (1-\epsilon_i) \rho'_i$ for $i = 1,2$
\State $\sigma^t_i = \left( 1 - \frac{1}{t+1} \right) \sigma^{t-1} _i + \frac{1}{t+1} \sigma'^t _i$ for $i = 1,2$
\State $\rho^t_i = \left( 1 - \frac{1}{t+1} \right) \rho^{t-1} _i + \frac{1}{t+1} \rho'^t _i$ for $i = 1,2$
\State $\tau^t_i = \left( 1 - \frac{1}{t+1} \right) \tau^{t-1} _i + \frac{1}{t+1} \tau'^t _i$ for $i = 1,2$
\EndFor
\State Output strategy profile $(\sigma^T_1,\sigma^T_2)$
\end{algorithmic}
\end{algorithm}

\begin{algorithm}
\caption{Approximation algorithm for $\epsilon$-safe equilibrium in $n$-player games, $n > 2$}
\label{al:approx-nplr}
\textbf{Inputs}: Game $G$, $\epsilon_i \in [0,1]$ for $i = 2,\ldots,n$, initial mixed strategy $\sigma^0_1 \in \Sigma_1$, initial mixed strategies $\tau^0_i, \rho^0_i \in \Sigma_i$ for $i = 2,\ldots,n$, number of iterations $T$.
\begin{algorithmic}
\State $\sigma^0_i = \epsilon_i \tau^0_i + (1-\epsilon_i) \rho^0_i$ for $i = 2,\ldots,n$
\For {$t = 1$ to $T$}
\State $\sigma'_1 = \argmax_{\sigma_1 \in \Sigma_1} u_1(\sigma_1,\sigma^{t-1}_{-1})$
\State $\rho'_i = \argmax_{\sigma_i \in \Sigma_i} u_i(\sigma_i,\sigma^{t-1}_{-i})$ for $i = 2,\ldots,n$
\State $\tau'_i = \argmin_{\sigma_i \in \Sigma_i} u_{1}(\hat{\sigma})$ where $\hat{\sigma}$ is the strategy profile where player $i$ follows $\sigma_i$ and the other players $j \neq i$ follow $\sigma^{t-1}_j$, for $j = 2,\ldots,n$
\State $\sigma'_i = \epsilon_i \tau'_i + (1-\epsilon_i) \rho'_i$ for $i = 2,\ldots,n$
\State $\sigma^t_i = \left( 1 - \frac{1}{t+1} \right) \sigma^{t-1} _i + \frac{1}{t+1} \sigma'^t _i$ for $i = 1, \ldots,n$
\State $\rho^t_i = \left( 1 - \frac{1}{t+1} \right) \rho^{t-1} _i + \frac{1}{t+1} \rho'^t _i$ for $i = 2,\ldots,n$
\State $\tau^t_i = \left( 1 - \frac{1}{t+1} \right) \tau^{t-1} _i + \frac{1}{t+1} \tau'^t _i$ for $i = 2,\ldots,n$
\EndFor
\State Output strategy profile $(\sigma^T_1,\ldots,\sigma^T_n)$
\end{algorithmic}
\end{algorithm}

\begin{proposition}
In Algorithms~\ref{al:approx-2plr} and~\ref{al:approx-nplr}, $\sigma^t_i = \epsilon_i \tau^t_i + (1-\epsilon_i) \rho^t_i$ for all $t$ and $i$ (for $i > 1$ for Algorithm~\ref{al:approx-nplr}). 
\end{proposition}
\begin{proof}
This is true for $t = 0$ by the definition of $\sigma^0_i.$ Now suppose 
$\sigma^t_i = \epsilon_i \tau^t_i + (1-\epsilon_i) \rho^t_i$ for all $t \leq k$, for some $k \geq 0$.
$$\sigma^{k+1}_i =  \left( 1 - \frac{1}{k+2} \right) \sigma^k _i + \frac{1}{k+2} \sigma'^{k+1} _i$$
$$= \left( 1 - \frac{1}{k+2} \right) \left(\epsilon_i \tau^k_i + (1-\epsilon_i) \rho^k_i\right) + \frac{1}{k+2} \left(\epsilon_i \tau'_i + (1-\epsilon_i) \rho'_i\right)$$
$$= \epsilon_i \tau^{k+1}_i + (1-\epsilon_i) \rho^{k+1}_i$$
\end{proof}


\vspace{-0.41in}
\section{Experiments}
\label{se:exp}
For the first set of experiments we investigate the runtime of our exact two-player algorithm as the number of pure strategies per player varies. For these experiments we set $\epsilon_1 = 0$, $\epsilon_2 = 0.05.$  We used an Intel Core i7-8550U at 1.80 GHz with 16 GB of RAM under 64-bit Windows 11 (8 threads). We used Gurobi version 9.5~\cite{Gurobi21:Gurobi}. We experimented on games with all payoffs uniformly random in [0,1]. For $m=2,3,5,10$ we ran 10,000 trials, and for $m = 15,20,25$ we ran 1,000. Here $m$ refers to the number of pure strategies per player (note that we experiment on games where all players have the same number of pure strategies, while our solution concepts and algorithms also apply to games where the players can have different numbers of pure strategies).
The results in Table~\ref{ta:results-2p-exact} indicate that the algorithm runs in less than a second for up to $m = 20$.

\begin{table}[!ht]
\centering
\begin{tabular}{|*{3}{c|}} \hline
$m$ &Avg. time(s) &Median time(s) \\ \hline
2 &\num{4.647e-4} &0.0\\ \hline
3 &0.001 &\num{9.975e-4}\\ \hline
5 &0.010 &0.007 \\ \hline
10 &0.062 &0.061  \\ \hline
15 &0.186 &0.173  \\ \hline
20 &0.736 &0.555 \\ \hline
25 &3.815 &2.007 \\ \hline
\end{tabular}
\caption{Running times for exact 2-player algorithm for varying number of pure strategies per player $(m)$, using $\epsilon_1 = 0$, $\epsilon_2 = 0.05$.}
\label{ta:results-2p-exact}
\end{table}

Next we experimented with the exact three-player algorithm, using $\epsilon_2 = \epsilon_3 = 0.05.$ Again we used Gurobi 9.5 with 8 cores on a laptop. For these experiments we used Gurobi's non-convex quadratic solver. For $m=2,3$ we ran 1,000 trials, and for $m = 4,5$ we ran 100. The results in Table~\ref{ta:results-3p-exact} indicate that the algorithm runs in a fraction of a second for $m = 2,3$ and several seconds for $m = 4.$

\begin{table}[!ht]
\centering
\begin{tabular}{|*{3}{c|}} \hline
$m$ &Avg. time(s) &Median time(s) \\ \hline
2 &0.036 &0.032\\ \hline
3 &0.194 &0.169\\ \hline
4 &4.856 &1.787\\ \hline
5 &468.731 &97.407 \\ \hline
\end{tabular}
\caption{Running times for exact 3-player algorithm for varying number of pure strategies per player $(m)$, using $\epsilon_1 = 0$, $\epsilon_2 = 0.05$, $\epsilon_3 = 0.05$.}
\label{ta:results-3p-exact}
\end{table}

We next experimented with our 2-player approximation algorithm (Algorithm~\ref{al:approx-2plr}). Again we used $\epsilon_1 = 0$, $\epsilon_2 = 0.05.$ For these experiments we just used a single core (note that the algorithm can be parallelized which would result in even lower runtimes). For each value of $m$ we ran 10,000 trials, performing 10,000 iterations of Algorithm~\ref{al:approx-2plr} for each trial. The results in Table~\ref{ta:results-2p-approx} indicate that the algorithm runs in just a fraction of a second for $m = 25.$ 

For player $i,$ define 
$$\delta^{\rho}_i = \max_{\sigma_i \in \Sigma_i} u_i(\sigma_i,\sigma^{T}_{-i}) - u_i(\rho^T_i,\sigma^{T}_{-i}).$$
That is, $\delta^{\rho}_i$ denotes the difference between the payoff of playing a best response to $\sigma^{T}_{-i}$ and following $\rho^T_i.$
Then define $\delta^{\rho} = \max_i \delta^{\rho}_i.$
Similarly, define 
$$\delta^{\tau} = u_{-i}(\sigma^{T}_{-i},\tau^T_i) - \min_{\sigma_i \in \Sigma_i} u_{-i}(\sigma^{T}_{-i},\sigma_i),$$
and $\delta^{\tau} = \max_i \delta^{\tau}_i.$
If both $\delta^{\rho} = 0$ and $\delta^{\tau} = 0$, then $\sigma^{T}$ would constitute an exact $\epsilon$-safe equilibrium.
So these values can be viewed as measures of approximation error. Table~\ref{ta:results-2p-approx} shows these approximation errors for different values of $m$.
While this algorithm runs significantly faster than the exact algorithm and can scale to larger games (particularly when implemented with parallelization), this 
comes at some cost to the accuracy of the solution.

\begin{table}[!ht]
\centering
\begin{tabular}{|*{4}{c|}} \hline
$m$ &Avg. time(s) &Avg. $\delta^{\rho}$ &Avg. $\delta^{\tau}$\\ \hline
2 &\num{9.642e-4} &\num{1.822e-4} &\num{9.385e-5} \\ \hline
3 &0.001 &\num{8.839e-4} &\num{6.865e-4} \\ \hline
5 &0.002 &0.004 &0.003 \\ \hline
10 &0.006 &0.013 &0.014\\ \hline
15 &0.010 &0.022 &0.027 \\ \hline
20 &0.016 &0.031 &0.039 \\ \hline
25 &0.026 &0.039 &0.049 \\ \hline
\end{tabular}
\caption{Running times and degrees of approximation error for 2-player approximation algorithm for varying number of pure strategies per player $(m)$, using $\epsilon_1 = 0$, $\epsilon_2 = 0.05$.}
\label{ta:results-2p-approx}
\end{table}

We next experimented with a variant of Algorithm~\ref{al:approx-2plr} based on a new initialization procedure for fictitious play called maximin initialization~\cite{Ganzfried22:Fictitious}.
While the prior experiments initialized $\rho_i$ and $\tau_i$ to be mixed strategies with equal probability for all pure strategies, maximin initialization generates a set of $K$ initial strategy profiles
and selects the run of the fictitious play algorithm that produces the smallest error. We can implement the same idea with Algorithm~\ref{al:approx-2plr}.
Let $\delta^{\rho,k}$ and $\tau^{\rho,k}$ denote the values of $\delta$ using initialization $k.$ 
Let $\delta'^k = \max\{\delta^{\rho,k}, \tau^{\rho,k}\},$ and let $k' = \argmin_k \delta'^k.$
Then define $\delta^{\rho} = \delta^{\rho,k'},$ $\delta^{\tau} = \delta^{\tau,k'}.$ For our experiments we used $K = 10$ initial strategy profiles. The results in Table~\ref{ta:results-2p-approx-maximin} show that maximin initialization significantly reduces the approximation error, though increases the runtime. We ran these experiments using a 64-core server parallelizing over the 10,000 trials, though we just used a single core for each algorithm run.

\begin{table}[!ht]
\centering
\begin{tabular}{|*{4}{c|}} \hline
$m$ &Avg. time(s) &Avg. $\delta^{\rho}$ &Avg. $\delta^{\tau}$\\ \hline
2 &0.186 &\num{1.029e-4} &\num{3.426e-5} \\ \hline
3 &0.243 &\num{4.775e-4} &\num{3.164e-4} \\ \hline
5 &0.327 &0.001 &0.001 \\ \hline
10 &0.370 &0.003 &0.003\\ \hline
15 &0.582 &0.003 &0.003 \\ \hline
20 &0.806 &0.004 &0.004 \\ \hline
25 &1.072 &0.004 &0.006 \\ \hline
\end{tabular}
\caption{Running times and degrees of approximation error for 2-player approximation algorithm for varying number of pure strategies per player $(m)$, using $\epsilon_1 = 0$, $\epsilon_2 = 0.05$,
using 10 initial strategy profiles.}
\label{ta:results-2p-approx-maximin}
\end{table}

We ran similar experiments for Algorithm~\ref{al:approx-nplr} on 3-player games, using $\epsilon_1 = 0$, $\epsilon_2 = 0.05$, $\epsilon_3 = 0.05$. Again we used a single core per run of the algorithm, and ran 10,000 trials for each value of $m$, with 10,000 iterations of the algorithm per trial. In Table~\ref{ta:results-3p-approx}, we define $\delta^{\rho}$ and $\delta^{\tau}$ as before except that they are just the maximum over the values for players 2 and 3. We also define
$$\delta^{\sigma} =  \max_{\sigma_1 \in \Sigma_1} u_i(\sigma_1,\sigma^{T}_{-1}) - u_1(\sigma^T_1,\sigma^{T}_{-1}).$$
The results in Table~\ref{ta:results-3p-approx} show how the running times and approximation errors for different values of $m$.
We also experimented on 3-player games using maximin initialization with $K = 10$. 
The results in Table~\ref{ta:results-3p-approx-maximin} show that, as for the two-player case, maximin initialization leads to a significant reduction in approximation error.

\begin{table}[!ht]
\centering
\begin{tabular}{|*{5}{c|}} \hline
$m$ &Avg. time(s) &Avg. $\delta^{\sigma}$ &Avg. $\delta^{\rho}$ &Avg. $\delta^{\tau}$\\ \hline
2 &0.007 &0.001 &0.002 &\num{6.996e-4} \\ \hline
3 &0.011 &0.003 &0.006 &0.004 \\ \hline
5 &0.021 &0.012 &0.017 &0.018 \\ \hline
10 &0.116 &0.037 &0.049 &0.057 \\ \hline
15 &0.322 &0.048 &0.063 &0.072 \\ \hline
20 &0.832 &0.057 &0.073 &0.080 \\ \hline
25 &1.707 &0.061 &0.077 &0.085 \\ \hline
\end{tabular}
\caption{Running times and degrees of approximation error for 3-player approximation algorithm for varying number of pure strategies per player $(m)$, using $\epsilon_1 = 0$, $\epsilon_2 = 0.05$,  $\epsilon_3 = 0.05$.}
\label{ta:results-3p-approx}
\end{table}

\begin{table}[!ht]
\centering
\begin{tabular}{|*{5}{c|}} \hline
$m$ &Avg. time(s) &Avg. $\delta^{\sigma}$ &Avg. $\delta^{\rho}$ &Avg. $\delta^{\tau}$\\ \hline
2 &0.414 &\num{5.224e-4} &\num{7.279e-4} &\num{2.930e-4} \\ \hline
3 &0.595 &0.002 &0.002 &0.002 \\ \hline
5 &0.993 &0.003 &0.004 &0.003 \\ \hline
10 &2.897 &0.007 &0.009 &0.010 \\ \hline
15 &7.706 &0.015 &0.019 &0.023 \\ \hline
20 &16.536 &0.023 &0.031 &0.035 \\ \hline
25 &36.392 &0.034 &0.040 &0.043 \\ \hline
\end{tabular}
\caption{Running times and degrees of approximation error for 3-player approximation algorithm using $\epsilon_1 = 0$, $\epsilon_2 = 0.05$,  $\epsilon_3 = 0.05$, and 10 initial strategy profiles.}
\label{ta:results-3p-approx-maximin}
\end{table}

\section{Conclusion}
We defined a new game-theoretic solution concept called safe equilibrium in which players behave potentially arbitrarily with some fixed probability $\epsilon_i.$ 
We proved that a safe equilibrium is guaranteed to exist for any number of players and all possible values of the parameters $\epsilon_i$, and we proved that its computation is PPAD-hard. We devised exact algorithms, both for the 2-player and $n$-player cases, which we showed are able to solve small games relatively quickly. We also presented approximation algorithms that achieve significantly lower runtimes but at the cost of a degree of approximation error. While we focused on strategic-form games, which model situations of perfect information with simultaneous actions, our analysis and algorithms should extend straightforwardly to more complex settings such as those with imperfect information and stochastic events. In the event that historical data of opponents' play is available, our algorithms can be integrated with an opponent modeling algorithm to provide models of opponents' irrational strategies that are more realistic than the worst-case assumption.

While Nash equilibrium has emerged as the central game-theoretic solution concept, its assumption that all players behave rationally may be too strict when modeling real human decision makers. As game theory is being increasingly applied to high-stakes situations, such as self-driving cars and national security, it is essential that strategies are able to accommodate the possibility of opponents' irrationality, which may be unpredictable. At the other end of the spectrum, a maximin strategy assumes that all opponents are trying to minimize our payoff, resulting in exceedingly conservative play with low payoffs. The new safe equilibrium concept effectively bridges the gap between these two extremes, enabling us to construct strategies that are robust to arbitrary degrees of opponents' irrationality.

\bibliographystyle{plain}
\bibliography{C://FromBackup/Research/refs/dairefs}

\end{document}